\documentclass[11pt]{article}

\usepackage{graphicx}
\usepackage{amssymb}
\usepackage{fancyhdr}
\usepackage{xcolor}
\usepackage{url}
\usepackage{mathtools}
\usepackage{amsmath}
\usepackage{adjustbox}
\usepackage{tabularx}
\usepackage[utf8]{inputenc}
\usepackage[english]{babel}
\usepackage{breakcites}
\usepackage{tabu}
\usepackage{mathrsfs}
\usepackage{dsfont}
\usepackage{booktabs}
\usepackage{amsthm}
\usepackage{hyperref}
\usepackage{rotating}
\usepackage{graphicx}
\usepackage{lscape}
\usepackage{enumitem}
\usepackage{bbm}
\usepackage[]{interval}
\usepackage{setspace}

\usepackage{geometry}
\intervalconfig{
  soft open fences,
  separator symbol =;,
}
\definecolor{cornellred}{rgb}{0.7, 0.11, 0.11}
\hypersetup{colorlinks,linkcolor={red},citecolor={cornellred},urlcolor={red}}

\newtheoremstyle{break}
  {\topsep}{\topsep}%
  {\itshape}{}%
  {\bfseries}{}%
  {\newline}{}%
  
\theoremstyle{break}
\newtheorem{definition}{Definition}

\newtheorem{theorem}{Theorem}

\newtheorem{fact}{Fact}
\newtheorem{assumption}{A}
\newtheorem{assumption_output}{AY}

\theoremstyle{definition}
\newtheorem{remark}{Remark}

\DeclareMathOperator{\spargel}{sp}
\DeclareMathOperator{\cspargel}{\overline{\spargel}}

\DeclareMathOperator{\diag}{diag}
\DeclareMathOperator{\rank}{rk}

\DeclareMathOperator{\E}{\mathbb E}
\DeclareMathOperator{\V}{\mathbb V}
\DeclareMathOperator{\Prob}{\mathbb P}
\DeclareMathOperator{\proj}{proj}

\DeclareMathOperator{\esssup}{ess\,sup}

\DeclareMathOperator{\mslim}{ms\,lim}

\renewcommand{\l}{\left}
\renewcommand{\r}{\right}

\newcommand\utimes{\mathbin{\ooalign{$\cup$\cr%
   \hfil\raise0.42ex\hbox{$\scriptscriptstyle\times$}\hfil\cr}}}
\newcommand\bigutimes{\mathop{\ooalign{$\bigcup$\cr%
   \hfil\raise0.36ex\hbox{$\scriptscriptstyle\boldsymbol{\times}$}\hfil\cr}}}

\makeatletter
\renewenvironment{proof}[1][\proofname]{%
  \par\pushQED{\qed}\normalfont%
  \topsep6\p@\@plus6\p@\relax
  \trivlist\item[\hskip\labelsep\bfseries#1\@addpunct{.}]%
  \ignorespaces
}{%
  \popQED\endtrivlist\@endpefalse
}
\makeatother

\DeclarePairedDelimiter\abs{\lvert}{\rvert}%
\DeclarePairedDelimiter\norm{\lVert}{\rVert}%

\makeatletter
\let\oldabs\abs
\def\abs{\@ifstar{\oldabs}{\oldabs*}}
\let\oldnorm\norm
\def\norm{\@ifstar{\oldnorm}{\oldnorm*}}
\makeatother

\usepackage[round]{natbib}

\makeatletter
\newcommand{\bianca}{\renewcommand\NAT@open{[}\renewcommand\NAT@close{]}}
\makeatother

\providecommand{\keywords}[1]{\textbf{\textit{Index terms---}} #1}

\providecommand{\keywords}[1]{\textbf{\textit{Index terms---}} #1}

\usepackage{stackengine}
\newcommand\ubar[1]{\stackunder[1.2pt]{$#1$}{\rule{.8ex}{.075ex}}}

\begin{document}
\onehalfspacing
%
\title{On the Existence of One-Sided Representations for the Generalised Dynamic Factor Model}
\author{
\textsc{Philipp Gersing}\footnote{Department of Statistics and Operations Research, University of Vienna, \url{philipp.gersing@univie.ac.at}},
}
\maketitle
%
%
%
%
%
%
\begin{abstract}

We show that the common component of the Generalised Dynamic Factor Model (GDFM) can be represented using only current and past observations basically whenever it is purely non-deterministic. 

\end{abstract}
%
%
%
%
%
%
%

%
%
%
%
%
%
%
%
%
%
%
%
%
\section{Introduction}
There are two main approaches to approximate factor models in time series: a) the dynamic approach, i.e., the Generalised Dynamic Factor Model (GDFM) based on dynamic principal components \citep{forni2000generalized, forni2001generalized}, and b) the static approach based on static principal components \citep{chamberlain1983arbitrage, chamberlain1983funds, stock2002forecasting, stock2002macroeconomic, bai2002determining}. Contrary to the common view in the literature, these are fundamentally different decompositions, each imposing distinct interpretations of what is ``common'' and ``idiosyncratic.'' For an in-depth discussion, see \cite{gersing2023reconciling, gersing2024weak} which was extended to the time domain by \cite{barigozzi2024dynamic}.

We say that a process is causally subordinated to the data if it can be expressed purely in terms of current and past values of observed variables. In the static approach, the \textit{static common component} is a linear combination of only contemporaneous observed variables, making it trivially causally subordinated. By contrast, the \textit{dynamic common component} of the GDFM is, in its original form \citep[see][]{forni2001generalized}, the mean-square limit of a dynamic low-rank approximation using lags \textit{and leads}.

This paper shows that the use of leads is only a matter of representation: under fairly general conditions, the dynamic common component can instead be written using only current and past variables. Specifically, we prove that its innovations (one-step ahead prediction errors) remain causally subordinated to the data, provided the component is purely non-deterministic and the transfer function meets a mild condition related to causal invertibility. Unlike low rank approximations via dynamic principal components in general, this \textit{one-sidedness} is a distinctive feature of the GDFM which is, as shown in this paper, implied by the special behaviour of its spectral eigenvalues. 

This result highlights why, from an economic perspective, the dynamic decomposition is of primary interest compared to the static decomposition. Interpreting innovations of the dynamic common component as ``common structural shocks of the economy'', the dynamic common component is the projection of observed variables onto the infinite past of those shocks \citep[see also][who interpret the dynamic idiosyncratic component as measurement error]{lippi2021validating, forni2025common}. Consequently, impulse response analysis in time series factor models should focus on how observed variables respond to these structural shocks and therefore be concerned with the dynamic common component. In contrast, the static decomposition captures only the part that is contemporaneously common.

In summary, the apparent two-sidedness of the classical GDFM is not an inherent flaw but a choice of representation. Our result reinforces the theoretical foundation of the GDFM by proving that, under mild conditions, the dynamic common component can always be represented in a causally subordinated, forecasting-relevant form.  This paper is intended as a theoretical contribution: (1) to establish the interpretation of the dynamic common component as the response to common structural shocks, and (2) to justify starting future research from a one-sided representation. 

The remainder of the paper is structured as follows. We begin in Section \ref{subsec: the main result and idea of the proof} by formally stating the main result and outlining the core idea of the proof. Next, Section \ref{sec: general setup} introduces the assumptions and notation underlying the GDFM. In Section \ref{sec: PND for infinite dimension}, we define purely non-deterministic processes in the infinite-dimensional, rank-deficient case and discuss aspects of causal invertibility. The main proof is presented in Section \ref{sec: one-sidedness of the common shocks}: we first show in Theorem \ref{thm: reordering to full rank blocks} how to construct infinitely many full-rank $q\times q$ transfer-function blocks, then establish causal subordination under a strict minimum phase condition in Theorem \ref{thm: causal subordination}, and finally relax this condition to include cases with the same zeros on the unit circle in infinitely many rows. The paper concludes with Section \ref{sec: conclusion}.

\subsection{The Main Result and Idea of the Proof}\label{subsec: the main result and idea of the proof}

To fix ideas, consider an infinite-dimensional time series as a double indexed (zero-mean, stationary) stochastic process $(y_{it}: i \in \mathbb N, t \in \mathbb Z)=(y_{it})$, indexed by cross-section $i \in \mathbb N$ and time $t \in \mathbb Z$. The GDFM decomposes \begin{align}
    y_{it} = \chi_{it} + \xi_{it} = \ubar b_i(L) u_t + \xi_{it} = \sum_{j = -\infty}^\infty B_i(j) u_{t-j} + \xi_{it}, \qquad u_t \sim WN(I_q) \label{eq: first GDFM in the world}
\end{align}
where $(u_t)$ is a $q$-dimensional orthonormal white noise process driving the dynamic common component $(\chi_{it})$ via square-summable filters $\ubar b_i(L)$, and $(\xi_{it})$ is the dynamic idiosyncratic component, weakly correlated over time and cross-section.

While the decomposition into common and idiosyncratic parts is unique, there are infinitely many equivalent representations of the filters and factor process. For any orthonormal $q\times q$ filter $\ubar c(L)$, we have
\begin{align}
\chi_{it} = \ubar b_i (L) \ubar c^*(L) \ubar c(L) u_t = \tilde{\ubar b}_i (L) \tilde u_t, \label{eq: pre and post mult rotation filters}
\end{align}
with $\tilde u_t = \ubar c(L) u_t$ being orthonormal white noise. The most straightforward way to estimate the GDFM is via dynamic principal components, leading to two-sided filters and factors \citep{forni2001generalized, forni2004generalized, hallin2007determining}, which cannot be used for forecasting. 

We show that whenever the common component is purely non-deterministic (plus a mild regularity condition) - a standard assumption in time series analysis — it admits a representation
\begin{align}
    \chi_{it} = \ubar k_i (L) \varepsilon_t = \sum_{j = 0}^\infty K_i (j) \varepsilon_{t-j}, \qquad \varepsilon_t \sim WN(I_q) \label{eq: first onsided rep of the world} 
\end{align}
where $\varepsilon_t \in \cspargel(y_{is} : i \in \mathbb N, s \leq t) := \mathbb H_t(y)$ and $\sum_{j = 0}^\infty \norm{K_i (j)}^2 <\infty$, where $\cspargel(\cdot)$ denotes the closed linear span. Here $(\varepsilon_t)$ is the innovation process of $(\chi_{it})$. This innovation form of the GDFM is unique (up to a real orthogonal matrix) and naturally one-sided.

Earlier work addressed the one-sidedness problem by imposing linear models on the dynamic common component: \cite{forni2005generalized} used a static factor structure with VAR dynamics, while \cite{forni2011general, forni2015dynamic, forni2017dynamic, barigozzi2024inferential} modeled the common component as VARMA, achieving one-sided representations in the shocks and output. This paper generalises these results extending the idea used in \cite{forni2015dynamic}: Write the GDFM in blocked form, with transfer-function blocks $\ubar k^{(j)}(L)$, $j = 1, 2,...$ of dimension $q\times q$: 
\begin{align}
    y_t &= \chi_t + \xi_t = \begin{pmatrix}
        \ubar k^{(1)}(L) & & \\
                         & \ubar k^{(2)}(L) & \\
                        && \ddots 
    \end{pmatrix} \begin{pmatrix}
        I_q \\
        I_q \\
        \vdots
    \end{pmatrix} \varepsilon_t + \xi_t  \nonumber \\[0.8em]
    \phi_t &:= \l(\ubar k(L)\r)^{-1} y_t  = \begin{pmatrix}
        I_q \\
        I_q \\
        \vdots
    \end{pmatrix} \varepsilon_t + \begin{pmatrix}
        \l(\ubar k^{(1)}(L)\r)^{-1} & & \\
                         & \l(\ubar k^{(2)}(L)\r)^{-1} & \\
                        && \ddots 
    \end{pmatrix} \xi_t \label{eq: phi preview}
\end{align}

If $(\chi_{it})$ is purely non-deterministic and $\ubar k(L)$ is the transfer-function of its Wold representation, the inverse of $\ubar k(L)$ is also causal. Note that $(\phi_t)$ resembles a static factor structure with $(\varepsilon_t)$ as its factors.  Suppose that all inverse transfer-functions $\l(\ubar k^{(j)}(L)\r)^{-1}$ are causal and such that the second term on the RHS of equation \eqref{eq: phi preview} is statically idiosyncratic. We can retrieve $\varepsilon_t$ causally from $(y_{it})$ by applying static principal components to $\phi_t$.

\section{General Setup}\label{sec: general setup}
\subsection{Notation}
%
%
%
Let $\mathcal P = (\Omega, \mathcal A, \Prob)$ be a probability space and $L_2(\mathcal P, \mathbb C)$ be the Hilbert space of square integrable complex-valued, zero-mean, random-variables defined on $\Omega$ equipped with the inner product $\langle u, v\rangle =  \E[u \bar v]$ for $u, v \in L_2(\mathcal P, \mathbb C)$. We suppose that $(y_{it})$ lives in $L_2(\mathcal P, \mathbb C)$ using the following abbreviations: $\mathbb H(y) := \cspargel(y_{it}: i \in \mathbb N, t \in \mathbb Z)$, the ``time domain'' of $(y_{it})$, $\mathbb H_t(y):= \cspargel(y_{is}: i \in \mathbb N, s \leq t)$, the ``infinite past'' of $(y_{it})$. We write $y_t^n = (y_{1t}, ..., y_{nt})'$ and by $f_y^n(\theta)$ we denote the ``usual spectrum'' of $(y_t^n)$ times $2\pi$, i.e., $\Gamma_y^n := \E\left[y_t^n (y_t^n)^*\right] = (2\pi)^{-1} \int_{-\pi}^\pi f_y^n(\theta) d\theta$. For a stochastic vector $u$ with coordinates in $L_2(\mathcal P, \mathbb C)$, we write $\V\l[u\r] := \E\left[u u^*\right]$ to denote the variance matrix. Let $u$ be a stochastic vector with coordinates in $L_2(\mathcal P, \mathbb C)$, let $\mathbb M\subset L_2(\mathcal P, \mathbb C)$ be a closed subspace. We denote by $\proj(u \mid \mathbb M)$ the orthogonal projections of $u$ onto $\mathbb M$ \citep[see e.g.][Theorem 1.2]{deistler2022modelle} (coordinate-wise). 
Furthermore we denote by $\mu_i(A)$ the $i$-th largest eigenvalue of a square matrix $A$. If $A$ is a spectral density $\mu_i(A)$ is a measurable function in the frequency $\theta \in [-\pi, \pi]$. More generally denote by $\sigma_i(A)$ the $i$-th largest singular value of a matrix (not necessarily square). 

\subsection{The Generalised Dynamic Factor Model}
Throughout we assume stationarity of $(y_{it})$ in the following sense:

\begin{assumption}[Stationary Double Sequence]\label{A: stat}
The process $(y^n_t : t \in \mathbb Z)$ is real valued, weakly stationary with zero-mean and such that
\begin{itemize}
    \item[(i)] $y_{it} \in L_2(\mathcal P, \mathbb C)$ for all $(i, t) \in \mathbb N \times \mathbb Z;$
    \item [(ii)] it has existing (nested) spectral density $f_{y}^n(\theta)$ for $\theta \in [-\pi, \pi]$ defined as the $n\times n$ matrix:
\[
f_{y}^n(\theta)=\frac 1{2\pi}\sum_{\ell=-\infty}^{\infty} e^{-\iota \ell \theta} 
\E [y_t^n y_{t-\ell}^{n'}],\quad  \theta \in[-\pi,\pi].
\]
\end{itemize}
\end{assumption}
In addition we assume that $(y_{it})$ has a $q$-dynamic factor structure as in \cite{forni2001generalized, hallin2011dynamic}: Denote by ``$\esssup$'' the essential supremum of a measurable function, we assume
\begin{assumption}[$q$-Dynamic Factor Structure]\label{A: q-DFS struct}
The process $(y_{it})$ is such that there exists $q < \infty$, with
\begin{itemize}    
    \item[(i)] $\sup_{n\in \mathbb N} \mu_q\left(f_y^n\right) = \infty$ almost everywhere on $[-\pi, \pi]$; 
    \item[(ii)] $\esssup_{\theta \in [-\pi, \pi]} \sup_{n\in \mathbb N} \mu_{q+1}(f_y^n)<  \infty$. 
\end{itemize}
\end{assumption}
By \cite{forni2001generalized} Assumption A\ref{A: q-DFS struct} is equivalent to the existence of the representation \eqref{eq: first GDFM in the world} with $(\xi_{it})$ and $(\chi_{it})$ being orthogonal at all leads and lags $\sup_{n\in \mathbb N}\mu_q\l(f_\chi^n(\theta)\r) = \infty$ almost everywhere on $[-\pi, \pi]$ and $\esssup_{\theta \in [-\pi, \pi]} \sup_{n\in \mathbb N} f_\xi^n(\theta) < \infty$. Note that this also implies that the shocks of $(\chi_{it})$ are orthogonal to $(\xi_{it})$ at all leads and lags and that $f_y^n(\theta) = f_\chi^n(\theta) + f_\xi^n(\theta)$ for $\theta$ almost everywhere in $[-\pi, \pi]$. 

We may also describe \eqref{eq: first GDFM in the world} as a representation rather than a ``model,'' since the existence of this dynamic decomposition follows from the characteristic eigenvalue behaviour of the double-indexed process 
$(y_{it})$. Specifically, the divergence of the first $q$ eigenvalues of $f_\chi^n$ captures the sense in which the filter loadings in \eqref{eq: first GDFM in the world} are pervasive. Meanwhile, the essential boundedness of $\sup_{n \in \mathbb N}\mu_1 \l( f_\xi^n \r)$ defines dynamic idiosyncraticness, ensuring that the dynamic idiosyncratic component is only weakly correlated across the cross-section and over time.

The approach of \cite{hallin2013factor} is even more general: the dynamic common component is defined as the projection onto the Hilbert space spanned by so-called ``dynamic aggregates,'' which arise as limits of weighted averages in the time domain \citep[see also][for the frequency domain version]{forni2001generalized}. The dynamic idiosyncratic component is then simply the residual from this projection. In principle, the GDFM framework requires only stationarity - without even assuming the existence of a spectral density - to state the decomposition directly. However, this also admits wired cases, such as $q = \infty$.
\section{Infinite Dimensional PND-Processes and Causal invertibility}\label{sec: PND for infinite dimension}
We begin with recalling some basic facts related to purely non-deterministic processes. Suppose for now that $(x_t) = (x_t^n)$ is a finite dimensional a zero-mean weakly stationary process. We call $\mathbb H^-(x) := \bigcap_{t\in \mathbb Z} \mathbb H_t(x)$ the remote past of $(x_t)$.
\begin{definition}[Purely Non-Deterministic and Purely Deterministic Stationary Process]\label{def: PND via remote past}
If $\mathbb H^-(x) = \{0\}$, then $(x_t)$ is called purely non-deterministic (PND) or regular. If $\mathbb H^-(x) = \mathbb H(x)$, then $(x_t)$ is called (purely) deterministic (PD) or singular.
\end{definition}
Note that ``singular'' in the sense of being purely deterministic must not to be confused with processes that have rank deficient spectrum, like the common component of the GDFM, and are also called ``singular'' in the literature \citep[][]{anderson2008generalized, deistler2010generalized, forni2024approximating}. Therefore we shall use PND and PD henceforth to avoid confusion. 

Of course there are processes between the two extremes of PD and PND. The future values of a PD process can be predicted perfectly (in terms of mean squared error). On the other hand, a PND process is entirely governed by random innovations. It can only be predicted with positive mean squared error and the further we want to predict ahead, the less variation we can explain: Set 
\begin{align*}
\nu_{h|t} &:= x_{t+h} - \proj(x_{t+h} \mid \mathbb H_t(x)), 
\end{align*}
which is the $h$-step ahead prediction error. If $(x_t)$ is PND, then $\lim_{h\to \infty} \V\l[\nu_{h|t}\r] \to \V\l[x_t\r]$. 

Next, we recall some basic facts about the finite dimensional case. By Wold's representation Theorem \citep[see][]{hannan2012statistical, deistler2022modelle}, any weakly stationary process can be written as sum of a PD and PND process, being mutually orthogonal at all leads and lags. 

There are several characterisations for PND processes. Firstly, the Wold decomposition implies \citep[see e.g.][]{rozanov1967stationary, masani1957prediction} that $(x_t)$ is PND if and only if it can be written as a causal infinite moving average
\begin{align}
    x_t = \nu_t + \sum_{j = 1}^\infty C(j)\nu_{t-j}, \label{eq: wold rep of x nu}
\end{align}
where $\nu_t:= \nu_{1|t-1}$ is the innovation of $(x_t)$, possibly with reduced rank $\rank \V\l[\nu_t\r] := q \leq n$. If $q<n$ we may uniquely factorise $\V\l[\nu_t\r] = bb'$. Assume without loss of generality that the first $q$ rows of $b$ have full rank (otherwise reorder), a unique factor is obtained choosing $b$ to be upper triangular with positive entries on the main diagonal. This results in the representation
\begin{align}
    x_t = \sum_{j = 0}^\infty K(j) \varepsilon_{t-j}, \label{eq: wold rep of x eps}
\end{align}
with $K(j) = C(j)b$ and $\V\l[\varepsilon_t\r] = I_q$. 

Secondly, a stationary process is PND \citep[see][]{rozanov1967stationary, masani1957prediction} if and only if the spectral density has constant rank $q\leq n$ almost everywhere on $[-\pi, \pi]$ and can be factored as  
\begin{align}
    &f_x(\theta) = 
    \underbrace{k(\theta)}_{n\times q} k^*(\theta) \label{eq: spectral factorisation x} \\
   \mbox{while} \quad & k(\theta) = \sum_{j = 0}^\infty K(j)e^{-\iota \theta }, \quad \sum_{j=0}^\infty \norm{K(j)}_F^2 < \infty,  \label{eq: k(z) of x} 
\end{align}
$\norm{\cdot}_F$ denotes the Frobenius norm and
\begin{align}
   k(\theta) = \ubar k(e^{-\iota \theta}),\quad  \theta \in (-\pi, \pi], \quad \ubar k(z) = \sum_{j=0}^\infty K(j)z^j, \quad z \in D, \label{eq: spectral factor of Wold for xt}
\end{align}
here $z$ denotes a complex number. The entries of the spectral factor $\ubar k(z)$ are analytic functions in the open unit disc $D$ and belong to the class $L^2(T)$, i.e., are square integrable on the unit circle $T$. If we employ the normalisation $K(0) = b$ from above, the transfer-function $\ubar k(z)$ from \eqref{eq: k(z) of x} which corresponds to the Wold representation \eqref{eq: wold rep of x eps} is causally invertible, i.e. $\mathbb H_t(\varepsilon) \subset \mathbb H_t(x)$. Causal invertibility is equivalent to $\rank k(z) = q$ for all $\abs{z} < 1$; there are no zeros inside the unit circle. We say also that the shocks $(\varepsilon_t)$ are \textit{fundamental} for $(x_t)$. 

\begin{fact}[\cite{szabados2022regular}]\label{fact: subselection}
If $(x_t)$ is PND with $\rank f_x = q < n$ almost everywhere on $[-\pi, \pi]$, then there exists a $q$-dimensional sub-vector of $x_t$, say $\tilde x_t = (x_{i_1,t}, ..., x_{i_q, t})'$ of full dynamic rank, i.e., $\rank f_{\tilde x} = q$ almost everywhere on $[-\pi, \pi]$.  
\end{fact}

To see why, we follow the proof of Theorem 2.1 in \cite{szabados2022regular}. A principal minor $M(\theta) = \det \left[(f_{x})_{i_j, i_l}\right]_{j, l = 1}^q$ of $f_x$ can be expressed by means of equation (\ref{eq: spectral factorisation x}) in terms of
\begin{align*}
    M_{f_x}(\theta) &= \det\left[\ubar k_{i_j, l}(e^{-\iota \theta})\right]_{j, l = 1}^q  \det \overline{\left[\ubar k_{i_j, l}(e^{-\iota \theta})\right]}_{j, l = 1}^q = \left|\det\left[\ubar k_{i_j, l}(e^{-\iota \theta})\right]_{j, l = 1}^q  \right|^2 := \left|M_{\underline{k}}(e^{-\iota \theta})\right|^2, 
\end{align*}
with the same row indices in the minor $M_{\underline{k}}(z)$ of $\ubar k(z)$ as in the principal minor $M(\theta)$ of $f_x$. We know that $M_{\underline{k}}(z) = 0$ almost everywhere or $M_{\underline{k}}(z) \neq 0$ almost everywhere because $M_{\underline{k}}(z)$ is analytic in $D$. Since $\rank f_x = q$ almost everywhere, the sum of all principal minors of $f_x$ of order $q$ is different from zero almost everywhere, so there exists at least one order $q$ principal minor of $f_x$ different almost everywhere from zero.

Let us now consider the case of an \textit{infinite dimensional rank deficient PND process} $(x_{it}: i \in \mathbb N, t \in \mathbb Z)$, so $\rank f_x^n = q$ almost everywhere on $[-\pi, \pi]$ for all $n \geq n_0$.
\begin{definition}[Purely non-deterministic rank-reduced stochastic double sequence]\label{def: character PND infinite dim}    
Let $(x_{it})$ be a stationary stochastic double sequence such that $\rank f_x^n = q <\infty$ almost everywhere on $[-\pi, \pi]$ for all $n\geq n_0$.  We say that $(x_{it})$ is PND if there exists an $n_1\geq n_0$ together with a $q$-dimensional orthonormal white noise process $(\varepsilon_t)\sim WN(I_q)$ such that
    \begin{itemize}
        \item[(i)]  $ \varepsilon_t \in \spargel\l(x_t^n -  \proj\l[x_t^n \mid \mathbb H_{t-1}(x^n)\r] \r)$ for all $n \geq n_1$ and $t\in \mathbb Z$; 
        \item[(ii)] $x_{it} \in \mathbb H_{t}(\varepsilon)$ for all $i\in \mathbb N$ and for all $i\in \mathbb N$ there is a causal transfer-function $\ubar k_i(L)$ such that 
            \begin{align}
                x_{it} = \ubar k_i (L) \varepsilon_t = \sum_{j = 0}^\infty K_i(j) \varepsilon_{t-j}, \label{eq: def innovation form singular}
            \end{align}
        where $\sum_{j = 0}^\infty \norm{K_i(j)}^2  < \infty$ and $\ubar k_i(z)$ are analytic in the open unit disc $D$ for all $i \in \mathbb N$. 
    \end{itemize}
\end{definition}
We conjecture that Definitions \ref{def: PND via remote past} and \ref{def: character PND infinite dim} are equivalent also for the infinite dimensional rank deficient case, which is however not the objective of the present paper. Uniqueness of $(\varepsilon_t)$ can be achieved e.g. by selecting the first index set in order such that the process has full rank almost everywhere on $[-\pi, \pi]$ and imposing constraints as described below equation \eqref{eq: wold rep of x nu}. An example for a purely deterministic double sequence would be $x_{it} = \varepsilon_{t+i-1}$; here we can perfectly predict the infinite future at time $t$ from $(x_{it}: i\in \mathbb N)$. 

Next, we discuss fundamentalness in the infinite dimensional, rank deficient case. Thinking of $(x_t) = (x_{1t}, x_{2t}, ....)'$ as an infinite dimensional vector process, a full rank $q$-dimensional sub-block as in Fact \ref{fact: subselection} has a transfer-function that is invertible, but not necessarily causally invertible. For illustration, consider the following examples with $\varepsilon_t$ being scalar ($q = 1$) white noise with unit variance:\\
\noindent
\begin{tabular}{@{}m{0.3\textwidth}@{\quad}m{0.3\textwidth}@{\quad}m{0.3\textwidth}@{}}
  \vspace{-\baselineskip} 
  \begin{equation}
  x_t = \begin{pmatrix}
        1-3L \\
        1-3L \\
        \vdots
    \end{pmatrix} \varepsilon_t 
    \label{eq: exmp not the Wold rep}
  \end{equation}
  &
  \vspace{-\baselineskip} 
  \begin{equation}
  x_t = \begin{pmatrix}
        1-0.5L \\
        \hline
        1-3L \\
        1-3L \\
        \vdots 
    \end{pmatrix} \varepsilon_t
    \label{eq: exmp invertible but not infinitely often}
  \end{equation}
  &
  \vspace{-\baselineskip} 
  \begin{equation}
  x_t = \begin{pmatrix}
        1-3L \\
        1-2L \\
        \hline
        1-3L \\
        1-2L \\
        \hline
        \vdots
    \end{pmatrix} \varepsilon_t.
    \label{eq: exmp zeros inside but not in 2 dimensions}
  \end{equation}
\end{tabular}
\\
Note that in all three cases $(x_{it})$ is a stationary PND double sequence. 
\begin{itemize}
    \item Starting with example \eqref{eq: exmp not the Wold rep}, let $\ubar k^n(z)$ be the transfer-function of $(x_t^n) = (x_{1t}, ..., x_{nt})'$ for $n\in \mathbb N$ as in \eqref{eq: spectral factor of Wold for xt}. We note that $\rank \ubar k^n(z_0) = 0$ for $z_0 = 1/3$ for all $n \in \mathbb N$. Therefore $\ubar k^n(z)$ has a zero inside the unit circle and is not causally invertible for any $n\in \mathbb N$ and \eqref{eq: exmp not the Wold rep} is not the Wold representation (a non-causal inverse representation is given by $\varepsilon_t = -1/3 \sum_{j = 1}^\infty (1/3)^{j-1} x_{1, t+j}$).
    
    By the spectral factorisation we can obtain a causally invertible factor of the spectrum of the univariate processes $x_{it}$ for $i \in \mathbb N$ by mirroring the zero on the unit circle: Rewrite $f(z) = (1 - 3z)(1-3z^{-1}) = (1-3z^{-1})z \ (1 - 3z)z^{-1} = 3(1 - 1/3z) 3(1-1/3 z^{-1})$. Consequently, there is a white noise unit variance scalar innovation process, say $(\eta_t)$ which is different from $(\varepsilon_t)$, such that $x_{it} = 3 \eta_t - \eta_{t-1}$ associated with the causally invertible transfer-function $3(1 - 1/3L)$. Here $(\eta_t)$ is the innovation for each individual univariate process and also for the entire multivariate infinite dimensional process $(x_t)$.
    \item In example \eqref{eq: exmp invertible but not infinitely often}, setting $\tilde x_t := x_{1t}$, the associated transfer-function is causally invertible. So is the transfer-function of any other sub-process of dimension $n > 1$ which includes the first coordinate $x_{1t}$. It follows that $(\varepsilon_t)$ is the innovation process of the multivariate process $(x_t)$.
    
    On the other hand setting $\tilde x_t := x_{it}$ for $i \geq 2$, the transfer-function of $(\tilde x_t)$ is not causally invertible. Hence, even though $(\varepsilon_t)$ is the innovation for the multivariate rank-deficient process $(x_t)$ it is in general not the innovation for its full rank sub-blocks (compare example \eqref{eq: exmp not the Wold rep}). The first coordinate settles the innovation for the whole infinite dimensional process $(x_t)$ and \eqref{eq: exmp invertible but not infinitely often} is the Wold representation.
    \item Finally in example \eqref{eq: exmp zeros inside but not in 2 dimensions}, the associated transfer-functions of all one-dimensional sub-processes are not causally invertible. However the transfer-functions of $2$-dimensional sub-blocks such as $\tilde x_t = (x_{1t}, x_{2t})'$ are causally invertible. They have full rank for all $z\in \mathbb C$ and therefore also for all $z$ inside the unit circle. For instance a causal inverse is given by $\varepsilon_t = -2 x_{1t} + 3 x_{2t}$.
    
    Consequently, potential non-fundamentalness of the shocks with respect to the output can be tackled by adding new cross-sectional dimensions which are driven by the same shocks, and therefore ``remove'' zeros inside the unit circle. For example \cite{forni2025common} exploit this fact to make structural VAR analysis more robust. As has been shown by \cite{anderson2016structure}, a rank deficient VARMA system (i.e. $n>q$) has an autoregressive representation, i.e. $\rank k^n(z) = q$ for all $z\in \mathbb C$ generically in the parameter space. For a different approach to non-fundamentalness see \cite{funovits2024identifiability}. 
\end{itemize}
To summarise, an infinite-dimensional, rank-deficient PND process typically contains many full-rank sub-blocks. While these sub-blocks are not necessarily causally invertible, they are ``more likely'' to be so as the dimension of the sub-block grows. This is because potential zeros inside the unit circle can be compensated for by the contribution of additional rows in the transfer function. 

Econometric time series analysis (in the realm of stationarity) is almost exclusively concerned with the modelling and prediction of ``regular'' time series. As noted above VAR, VARMA and state space models are all PND \citep[see][]{deistler2022modelle}. If we think of $(y_{it})$ as a process of (stationarity transformed) economic data, we would not expect that any part of the variation of the process could be explained in the far distant future given information up to now. The same should hold true for the common component which explains a large part of the variation of the observed process. Even more so, if we interpret the idiosyncratic component of the GDFM as measurement errors \cite{lippi2021validating, forni2025common}.

Therefore, we impose the following assumption:
\begin{assumption}[Purely Non-Deterministic Dynamic Common Component]\label{A: dynamic CC is PND}
The dynamic common component $(\chi_{it})$ of the GDFM is PND with orthonormal white noise innovation $(\varepsilon_t)$ (of dimension $q$) and innovation-form 
\begin{align*}
\chi_{it} =\ubar k_i (L)\varepsilon_t = \sum_{j = 0}^\infty K_i (j) \varepsilon_{t-j}.
\end{align*}
\end{assumption}
This assumption resolves only half of the one-sidedness issue as it is not clear whether $\varepsilon_t$ has a representation in terms of current and past $y_{it}$'s.
%
%
%
%
%
\section{One-Sidedness of the Common Shocks in the Observed Process}\label{sec: one-sidedness of the common shocks}
Consider the sequence of $1\times q$ row transfer-functions $(k_i: i \in \mathbb N)$. For the proof of Theorem \ref{thm: causal subordination}, the main result of the paper, we rely on the following key property: By reordering and stacking, we can construct a sequence of blocks $(k^{(j)}: j\in \mathbb N)$ of dimension $q_j\geq q$ such that the left-inverses $(k^{(j)})^{\dag}$  (here ``$\dag$'' denotes the generalised inverse) are causal and absolutely summable filters. First, we show that we can build infinitely many full-rank blocks by reordering the sequence. Second, we argue why it is reasonable to assume that we can stack the blocks in such a way that each block is also \textit{causally} invertible, i.e. there are no zeros \textit{inside} the unit circle. Third, absolute summability requires that the blocks have no zeros \textit{on} the unit circle - a condition we will relax in the discussion following the proof of Theorem \ref{thm: causal subordination}.

\begin{theorem}\label{thm: reordering to full rank blocks}
Under Assumptions A\ref{A: stat}-A\ref{A: dynamic CC is PND}, there exists a reordering $\left(k_{i_l}: l \in \mathbb N\right)$ of the sequence $(k_{i}: i \in \mathbb N)$ such that all consecutive $q \times q$ blocks $(k^{(j)})$ of $\left(k_{i_l}: l \in \mathbb N\right)$ have full rank $q$ almost everywhere on $[-\pi, \pi]$.
\end{theorem}

\begin{remark}\label{rem: no zero rows}
By Assumption A\ref{A: q-DFS struct}, we know that 
\begin{align}
    \mu_q\left(f_\chi^n \right) = \mu_q \bigg( \left(k^n\right)^* k^n \bigg) \to \infty \quad \mbox{almost everywhere on} \ [-\pi, \pi], \label{eq: divergence k^n* k^n}  
\end{align}
with $k^n = (k_1', ...., k_n')'$.

Since by Assumption A\ref{A: dynamic CC is PND}, $\ubar k_i$ is analytic in the open unit disc, it follows that either $k_i(\theta) = 0$ or $k_i(\theta) \neq 0$ almost everywhere on $[-\pi, \pi]$. If $k_i(\theta) = 0$ almost everywhere, then $\chi_{it} = 0$ and therefore $\chi_{it} \in \mathbb H_t(y)$. By (\ref{eq: divergence k^n* k^n}) the number of non-zero rows $k_i$ must be infinite. Therefore Theorem \ref{thm: reordering to full rank blocks} holds if and only if it holds after removing all rows with $k_i = 0$ almost everywhere on $[-\pi, \pi]$.
\end{remark}
\begin{proof}[Proof of Theorem \ref{thm: reordering to full rank blocks}]

Concurring with Remark \ref{rem: no zero rows}, we assume that $(k_i: i \in \mathbb N)$ has no zero rows and prove the statement by constructing the reordering using induction. By Assumption A\ref{A: dynamic CC is PND} and Fact \ref{fact: subselection} and equation (\ref{eq: divergence k^n* k^n}), we can build the first $q\times q$ block, having full rank almost everywhere on $[-\pi, \pi]$ by selecting the first linearly independent rows $i_1, ..., i_q$ of the sequence of row transfer-functions $(k_i: n \in \mathbb N)$, i.e., set $k^{(1)} = (k_{i_1}', ... k_{i_q}')'$. 

Now look at the block $j+1$: We use the next $k_i$ available in order, as the first row of $k^{(j+1)}$, i.e., $k_{i_{jq + 1}}$. Suppose we cannot find $k_i$ with $i \in \mathbb N \setminus \{ i_l: l \leq jq +1 \}$ linearly independent of $k_{i_{jq + 1}}$. Consequently, having built already $j$ blocks of rank $q$, all subsequent blocks that we can build from any reordering are of rank $1$ almost everywhere on $[-\pi, \pi]$. In general, for $\bar q < q$, suppose we cannot find rows $k_{i_{jq + \bar q + 1}}, ..., k_{i_{jq + q}}$ linearly independent of $k_{i_{jq+1}}, ...,k_{i_{jq + \bar q}}$, then all consecutive blocks that we can obtain from any reordering have at most rank $\bar q$.

For all $m = j+1, j+2, ...$ by the RQ-decomposition we can factorise $k^{(m)} = R^{(m)} (\theta) Q^{(m)}$, where $Q^{(m)} \in \mathbb C^{q \times q}$ is orthonormal and $R^{(m)}(\theta)$ is lower triangular $q\times q$ filter which is analytic in the open unit disc. For $n \geq i_{j}$ and without loss of generality that $n$ is a multiple of $q$, the reordered sequence looks like
{
\small
\begin{align*}
     (k^n)^* k^n &= \sum_{l=1}^n k_{i_l}^* k_{i_l} \\[-1em]
     &=\begin{bmatrix}
        \left(k^{(1)}\right)^* & \cdots & \left(k^{(j)}\right)^*
    \end{bmatrix} 
    \begin{bmatrix}
         k^{(1)} \\
        \vdots\\ 
        k^{(j)} 
    \end{bmatrix}
    + 
    \begin{bmatrix}
        \left(R^{(j+1)}\right)^* & \cdots & \left(R^{(J)}\right)^*
    \end{bmatrix}
      \begin{bmatrix}
        R^{(j+1)} \\
        \vdots \\
        R^{(J)}
    \end{bmatrix} \\
    &= \begin{bmatrix}
        \left(k^{(1)}\right)^* & \cdots & \left(k^{(j)}\right)^*
    \end{bmatrix} 
    \begin{bmatrix}
         k^{(1)} \\
        \vdots\\ 
        k^{(j)} 
    \end{bmatrix} + \begin{pmatrix}
        \times & 0 \\
           0    & 0
    \end{pmatrix}
    = A + B^n, \ \mbox{say,} 
\end{align*}
}
where $\times$ is a placeholder. By the structure of the reordering, there are $q - \bar q$ zero end columns/rows in $B^n$ for all $n \geq jq$ where $A$ remains unchanged.

Now by \citet[][theorem 1, p.301]{lancaster1985theory}, we have
\begin{align*}
    \mu_q\bigg(\left(k^n\right)^* k^n\bigg) &= \mu_q(A + B^n) \\[-0.8em]
    &\leq \mu_1(A) + \mu_q (B^n) \\
    &= \mu_1(A)  < \infty \ \mbox{for all } n \in \mathbb N \ \mbox{almost everywhere on} \ [-\pi, \pi] . 
\end{align*}
This also implies that for \textit{any} reordering the $q$-th eigenvalue of the resulting inner product of the transfer-function as in equation (\ref{eq: divergence k^n* k^n}) is bounded by $\mu_1(A)$. This is a contradiction and completes the induction step and the proof.  
\end{proof}
Theorem \ref{thm: reordering to full rank blocks} shows that we can extract infinitely many sub-blocks of of dimension $q$ from $(\chi_{it})$, each with full-rank spectrum almost everywhere on $[-\pi, \pi]$. This provides us with an unlimited supply of such ``variable stacks,'' each capable of capturing the signal from all $q$ common shocks.

Next, we impose the following assumption:
\begin{assumption}[Uniformly Strictly Minimum Phase after Blocking]\label{A: strict miniphase}
    There exists a sequence of blocks $(k^{(j)}: j \in \mathbb N)$ of dimension $q_j\times q$ with $q_j\geq q$ constructed from $(k_i: i \in \mathbb N)$ (by reordering/elimination and appropriate blocking), such that $\sigma_q\left(k^{(j)}\right) > \delta > 0$ almost everywhere on $[-\pi, \pi]$ for all $j \in \mathbb N$. 
\end{assumption}
This is similar to the commonly employed assumption in linear systems theory that the transfer-function is strictly minimum-phase, i.e., has no zeros on the unit circle as assumed in \cite{deistler2010generalized}, section 2.3 or in \cite{forni2017dynamic}, Assumption 7. 

Some further comments in order: Consider a sequence of full rank $q\times q$ blocks as in Theorem \ref{thm: reordering to full rank blocks}. First let us remark that the lack of causal invertibility of a full rank transfer-function block is the non-standard case. So we could simply assume that all consecutive $q\times q$ blocks (or a subsequence thereof) are causally invertible. On the other hand this excludes examples \eqref{eq: exmp invertible but not infinitely often} and \eqref{eq: exmp zeros inside but not in 2 dimensions}. However, as demonstrated in the discussion of example \eqref{eq: exmp zeros inside but not in 2 dimensions}, zeros can be removed by adding additional linearly independent rows to a block. For instance, we might be able to paste blocks together, increasing to dimension to $q_j\times q$ with $q_j \geq q$, while $q_j$ can be arbitrarily large, such that all zeros inside the unit circle vanish, so Assumption A\ref{A: strict miniphase} holds for \eqref{eq: exmp zeros inside but not in 2 dimensions}. Furthermore, note that example \eqref{eq: exmp not the Wold rep} is also covered by A\ref{A: strict miniphase} with $(\eta_t)$ as innovation instead of $(\varepsilon_t)$ (see the related discussion). Only cases like \eqref{eq: exmp invertible but not infinitely often} are ruled out by A\ref{A: strict miniphase}, where the innovation is determined by a finite number of transfer-function rows while all other rows have the same zeros inside the unit circle which cannot be removed by stacking. Even then, causal invertibility holds if we would ignore those rows, i.e. row $i = 1$ in \eqref{eq: exmp invertible but not infinitely often} which brings us back to the case of \eqref{eq: exmp not the Wold rep}.

Next, Assumption A\ref{A: strict miniphase} requires the stacks to be not only minimum phase (no zeros inside the unit circle) but \textit{strictly} minimum phase (no zeros inside and on the unit circle) with a uniform bound $\delta$. This excludes e.g. $\ubar k_i(L) = 1-L, i\in \mathbb N$ or $\ubar k_i(L) = (1 - (1-\frac{1}{i})L), i\in \mathbb N$. We will show how to incorporate those cases after the proof of Theorem \ref{thm: causal subordination}. On the other hand, we may argue that the uniformly strict minimum phase property with a global bound $\delta$ can be achieved by eliminating blocks or extending their size as described above. 

In summary, we conclude that Assumption A\ref{A: strict miniphase} is in fact a very mild restriction, excluding only rather contrived edge cases.
\begin{theorem}\label{thm: causal subordination}
Suppose A\ref{A: stat}-A\ref{A: strict miniphase} hold for $(y_{it})$, then the innovations $(\varepsilon_t)$ of $(\chi_{it})$ are causally subordinated to the observed variables $(y_{it})$, i.e., $\varepsilon_t \in \mathbb H_t(y)$. 
\end{theorem}
We apply remark \ref{rem: no zero rows} also for Theorem \ref{thm: causal subordination}. Trivially, since $(k_i: i \in \mathbb N)$ are also causal, Theorem \ref{thm: causal subordination} directly implies that the dynamic common component is causally subordinated to the observed output, i.e. $\chi_{it} \in \mathbb H_t(y)$ for all $i\in \mathbb N, t \in \mathbb Z$. Furthermore, if we had to eliminate rows to satisfy Assumption A\ref{A: strict miniphase}, then after recovering the common shocks $(\varepsilon_t)$ causally from $(y_{it})$, we can also reconstruct the common component of the eliminated variables by projection $\chi_{it} = \proj(y_{it} \mid \mathbb H_t(\varepsilon_t))$ \citep[see][]{forni2001generalized, gersing2023reconciling}.

\begin{proof}[Proof of Theorem \ref{thm: causal subordination}]
Suppose $(k^{(j)} : i \in \mathbb N)$ is such that Assumption A\ref{A: strict miniphase} is satisfied. Suppose $\sum_{j = 1}^J q_j = n$ without loss of generality.
{
\begin{align*}
\chi_t^n =
\begin{pmatrix}
\chi_t^{(1)} \\ \chi_t^{(2)} \\ \vdots \\ \chi_t^{(J)}
\end{pmatrix} 
= \begin{pmatrix} 
\ubar k^{(1)}(L) \\
\vdots \\
\ubar k^{(J)}(L) 
\end{pmatrix} \varepsilon_t & =  
\begin{pmatrix}
 \ubar k^{(1)}(L) & &\\
& \ddots & \\
& & \ubar k^{(J)}(L)
\end{pmatrix} 
\begin{pmatrix}
 I_q \\
 \vdots \\
 I_q
\end{pmatrix} \varepsilon_t .
\end{align*}
}
By Assumption A\ref{A: strict miniphase}, we know that all left-inverse transfer-functions $(k^{(j)})^{\dag}, j = 1, ..., J$ are causal as well. Next we show that $(\varphi_{it})$ in 
{
%
\begin{align}
\varphi_t^{qJ} &:= \begin{pmatrix}
   \left(\ubar k^{(1)}\right)^{\dag}(L) & & \\
    & \ddots & \\
    & &  \left(\ubar k^{(J)}\right)^{\dag}(L)
\end{pmatrix} 
\begin{pmatrix}
y_t^{(1)} \\ \vdots \\ y_t^{(J)}
\end{pmatrix} \nonumber \\
&=
\begin{pmatrix}
 I_q \\
 \vdots \\
 I_q
\end{pmatrix} \varepsilon_t +
\begin{pmatrix}
  \left(\ubar k^{(1)}\right)^{\dag}(L) & &\\
& \ddots & \\
& &  \left(\ubar k^{(J)}\right)^{\dag}(L)
\end{pmatrix}
\begin{pmatrix}
\xi_t^{(1)} \\ \vdots \\ \xi_t^{(J)}
\end{pmatrix} 
= C_t^{\varphi, qJ} + e_t^{\varphi, qJ} , \mbox{say,} \label{eq: varphi and averaging}  \  
\end{align}
}
has a static factor structure (see definition \ref{def: r-SFS}). 
Then $\varepsilon_t$ can be recovered from static aggregation/ via static principal components applied to $(\varphi_{it})$ by Theorem \ref{thm: charact r-SFS}.2.

Firstly, $q$ eigenvalues of $\Gamma_{C^\varphi}^{qJ} = \E \left[C_t^{\varphi, qJ} C_t^{\varphi, qJ'}\right]$ diverge for $J(n) \to \infty$ as $n\to \infty$, so A\ref{A: r-SFS struct}(i) holds. We are left to show that the first eigenvalue of $\Gamma_{e^\varphi}^n = \E\left[e_t^{\varphi, qJ}e_t^{\varphi, qJ'}\right]$ is bounded in $qJ$, i.e., A\ref{A: r-SFS struct}(ii) holds. Let $U_j \Sigma_j V_j^* = k^{(j)}(\theta)$ be the singular value decomposition while $U_j$ is $n\times q$ with orthonormal columns, $\Sigma_j = \diag(\sigma_1(k^{(j)}), ..., \sigma_q(k^{(j)}))$ is the diagonal matrix of singular values of $k^{(j)}$ and $V_j$ is a $q\times q$ unitary matrix, where we suppressed the dependence on $\theta$ in the notation on the LHS. Let $f_\xi^n(\theta) = P^* M P$ be the eigen-decomposition of $f_\xi^n$ with orthonormal eigenvectors being the rows of $P$ and eigenvalues in the diagonal matrix $M$ (omitting dependence on $n$). Then 
\begin{align*}
   f_{e^\varphi}^{qJ}(\theta) =  \bigoplus_{j = 1}^J V_j \underbrace{\bigoplus_{j = 1}^J \Sigma_j^{-1} \bigoplus_{j = 1}^J U_j^* P^* M P \bigoplus_{j = 1}^J U_j \Sigma_j^{-1}}_{B^J(\theta)} \bigoplus_{j = 1}^J V_j^*, 
\end{align*}
where we used $\bigoplus_{j = 1}^J A_j$ to denote the block diagonal matrix with the square matrices $A_j$ for $1\leq j \leq J$ on the main diagonal block. The largest eigenvalue of $f_{e^\varphi}^{qJ}(\theta)$ is equal to the largest eigenvalue of $B^J(\theta)$. Therefore by Jensen's inequality and Assumption A\ref{A: strict miniphase} we have
\begin{align*}
\mu_1\left(\Gamma_{e^\varphi}^{qJ} \right) &= \mu_1\left(\int_{-\pi}^\pi f_{e^{\varphi}}^{qJ}\right) \leq \int_{-\pi}^\pi \mu_1\left(f_{e^{\varphi}}^{qJ}\right)  \\[0.8em]
&\leq 2\pi \sup_{J\in \mathbb N}  \esssup_{\theta \in [-\pi, \pi]} \mu_1 \left(f_{e^{\varphi}}^{qJ}\right) \\
& \leq 2 \pi \sup_{J\in \mathbb N} \l\{\esssup_{\theta \in [-\pi, \pi]}  \mu_1  \left(f_{\xi}^{n}\right) \times \sup_{1\leq j \leq J} \esssup_{\theta \in [-\pi, \pi]} \sigma_q\l(k^{(j)}\r)^{-2} \r\} \\
& \leq 2 \pi \sup_{J\in \mathbb N} \esssup_{\theta \in [-\pi, \pi]}  \mu_1  \left(f_{\xi}^{n}\right) \delta^{-2}  < \infty.  
\end{align*}
This completes the proof.
\end{proof}
Note that we employed the uniform strict minimum phase property of Assumption A\ref{A: strict miniphase} in the proof of Theorem \ref{thm: causal subordination} to ensure that the filtered idiosyncratic blocks $(k^{(j)})^{\dag}(L) \xi_t^{(j)} = e_t^{\varphi, (j)}$ have finite variance. This condition prevents the highly-nongeneric situation in which zeros appear on the unit circle \textit{in almost all} blocks, irrespective of how they are stacked. Nevertheless, in what follows, we show that absolute summability of the left-inverse transfer function blocks is not in fact required causal subordination: 

For instance, consider the model
\begin{align*}
y_{it} = \chi_{it} + \xi_{it} =  (1-L) \varepsilon_t + \xi_{it} = \zeta_t + \xi_{it}
\end{align*}
where $(\xi_{it})$ is dynamically idiosyncratic. Clearly, the coefficients of $(1-L)^{1}$ are not absolutely summable. Still, we can recover $(\varepsilon_t)$ one-sided in the $(y_{it})$: The cross-sectional average $\bar y_t^n = n^{-1}\sum_{i = 1}^n y_{it}$ is a static aggregation of $(y_{it})$ and therefore $\bar y_t^n \to \zeta_t$ converges in mean square \citep[see][for details and the appendix for a short summary]{gersing2023reconciling, gersing2024distributed}. Furthermore $(\zeta_t)$ is PND and by the Wold representation Theorem, the innovations are recovered from the infinite past $\mathbb H_t(\zeta) \subset \mathbb H_t(y)$: It suffices that the inverse of $(1-L)$ exist for the input $(\zeta_t)$, since 
\begin{align*}
    &f_\zeta(\theta) = \left(1-e^{-\iota \theta }\right) f_\varepsilon(\theta) \left(1-e^{\iota \theta }\right) = \left(1-e^{-\iota \theta }\right) \frac{1}{2\pi} \left(1-e^{\iota \theta }\right)  \\
    &\int_{-\pi}^\pi  \left(1-e^{-\iota \theta }\right)^{-1} f_\zeta(\theta) \left(1-e^{\iota \theta }\right)^{-1} = \int_{-\pi}^\pi \frac{1}{2\pi} = 1 < \infty, 
\end{align*}
we know that $(1-e^{-\iota \theta})^{-1}$ is an element of the frequency domain of $(\zeta_t)$ and therefore has an inverse also in the time domain \citep[see also][section 5]{anderson2016multivariate}. 

More generally, we may employ this procedure by factoring out the zeros from the analytic functions $(\ubar k^{(j)}: j\in \mathbb N)$. Let $\ubar k^{(j)} = \ubar g_j \ubar h_j$, where $\ubar g_j$ is a polynomial defined by the zeros of $\ubar k^{(j)}$ which are on the unit circle. Recall that the zeros of an analytic function are isolated, so if $z_0$ is a zero, we have $\ubar h_j(z)\neq 0$ in a neighbourhood around $z_0$. Furthermore the degree of a zero can be only finite or the function is zero everywhere. Thus there can be only finitely many different zeros on the unit circle, since the unit circle is a compact set. Write $\ubar g_j(z) = \Pi_{k_j = 1}^{M_j} (z - z_{j k_j})^{m_{k_j}}$ with $|z_{j k_j}| = 1$ for $k_j = 1, ..., M_j$ and $j \in \mathbb N$. It follows that $g_j^{-1}(\theta) k^{(j)}(\theta) \neq 0$ almost everywhere on $[-\pi, \pi]$, where $g_j(\theta) := \ubar g_j (e^{-i\theta})$. 

Consequently setting $\ubar k^{(j)}(L) = \ubar g_j(L) \ubar h^{(j)}(L)$, we have 

%
%
\begin{align}
\varphi_t^n = \begin{pmatrix}
 \ubar g_1(L)  I_q \\
 \vdots \\
 \ubar g_{J}(L) I_q
\end{pmatrix} \varepsilon_t +
\begin{pmatrix}
  \left(\ubar h^{(1)}\right)^{\dag}(L) & &\\
& \ddots & \\
& &  \left(\ubar h^{(J)}\right)^{\dag}(L)
\end{pmatrix}
\begin{pmatrix}
\xi_t^{(1)} \\ \vdots \\ \xi_t^{(J)}
\end{pmatrix} = C_t^{\varphi, n} + e_t^{\varphi, n},  \label{eq: varphi and averaging more general} 
\end{align}
with $\rank \ubar h_j(L) =  q$ almost everywhere on $[-\pi, \pi]$. Now, if Assumption A\ref{A: strict miniphase} holds for $(h^{(j)})$ instead of $(k^{(j)})$, with the same arguments as above, we obtain a one-sided representation of $(\varepsilon_t)$, if there exists a static averaging sequence $(\hat c_i^{(n)}:(i, n) \in \mathbb N \times \mathbb N)$ (see definition \ref{def: SAS}) such that
\begin{align}
    \zeta_t^n = \sum_{i = 1}^n c_i^{(n)}\varphi_{it} = \sum_{i = 1}^n c_i^{(n)}(C_{it}^{\varphi} + e_{it}^{\varphi}) \to \zeta_t \label{eq: zeta aggregate for retrieval}
\end{align}
converges in mean square to a PND process, say $\zeta_t$, with innovations $(\varepsilon_t)$, with $\varphi_{it}, C_{it}^{\varphi}, e_{it}^\varphi$ from equation (\ref{eq: varphi and averaging more general}).

Summing up, even in the highly non-generic case where zeros lie on the unit circle in almost all transfer-function blocks, regardless of how they are stacked, it remains possible to retrieve the common innovations $(\varepsilon_t)$ causally from $(y_{it})$ by factoring out the zeros, inverting the invertible part, and then aggregating to recover the common shocks. This requires the existence of an aggregate $(\zeta_t)$ as in \eqref{eq: zeta aggregate for retrieval} with innovations $(\varepsilon_t)$. As formally proving this in complete generality may be challenging, we assume such pathological cases are of limited practical relevance and do not pursue them further here.

The inclusion of edge cases relating zeros inside and on the unit circle in the $q\times q$ transfer-function blocks, while mostly of theoretical interest, highlight the generality of the GDFM’s one-sidedness rather than suggesting a practical estimation method. In practice, achieving the required structure by reordering and stacking variables is non-trivial, but we may assume that robust approaches - like using blocks of size $q+1$ or $q+2$ as in \cite{forni2015dynamic, forni2017dynamic, barigozzi2024inferential} — are sufficient. Alternatively, recent methods \cite{gersing2024distributed, gersing2024weak} estimate the dynamic common component by projecting onto current and past factors (imposing additional assumptions) extracted via static principal components, avoiding concerns about non-fundamentalness or unit-circle zeros — provided a one-sided representation exists.

\section{Conclusion}\label{sec: conclusion}
We conclude by highlighting that causal subordination is a distinctive feature of the GDFM decomposition, rooted in the specific behaviour of its diverging spectral eigenvalues. Unlike general dynamic low-rank approximations via dynamic principal components, e.g. of dimension $q+h$, for $h>1$, the GDFM’s $q$ diverging spectral eigenvalues allow us to construct infinitely many full-rank transfer-function blocks. By causally inverting these blocks (potentially after factoring out spectral zeros) and then aggregating, we can recover the common innovations causally from the observed process. Consequently, provided that the dynamic common component is purely non-deterministic, it is itself causally subordinated to the observed process.
\section*{Acknowledgements}
{{I am deeply grateful to my PhD supervisor, Manfred Deistler, for his guidance and support. I would also like to thank the two anonymous referees for their valuable comments, particularly for highlighting the issue of causal invertibility, which significantly improved the paper. My thanks further go to Paul Eisenberg and Sylvia Frühwirth-Schnatter for their helpful suggestions. Financial support from the Austrian Central Bank under Anniversary Grant No.18287, the DOC Fellowship of the Austrian Academy of Sciences (ÖAW), and the University of Vienna is gratefully acknowledged.
}}
\section*{Data Availability Statement}
Data sharing is not applicable to this article as no datasets were generated or analysed during the current study.
\bibliographystyle{apalike} 
\bibliography{references.bib}
\appendix
\renewcommand{\theequation}{\thesection.\arabic{equation}}
%
%
%
\section{Background: Hilbert Space Theory for the Static Case}
Consider infinite dimensional \textit{constant} row-vectors of cross-sectional weights $\hat c  = (\hat c_1, \hat c_2, \cdots )\in \mathbb R^{1\times \infty}$ and write $\hat c^{\{n\}}:= (\hat c_1, \cdots \hat c_n)$ for the truncated vector. Denote by $\hat L_2^\infty (\Gamma_y^n)$ the set of vectors such that $\lim_{n\to \infty} \hat c^{\{n\}} \Gamma_y^n {\left(\hat c^{\{n\}}\right)}' < \infty$, where $\Gamma_y^n = \E\left[y_t^n y_t^{n'}\right]$ and by $\hat L_2^\infty(I)$ the set of vectors such that $\lim_{n\to \infty} \hat c^{\{n\}} {(\hat c^{\{n\}})}' < \infty$. Dynamic averaging sequences have been introduced by \cite{forni2001generalized}. In \cite{gersing2023reconciling} those are paralleled with static averaging sequences. For an alternative averaging scheme see \cite{barigozzi2024dynamic}.
\begin{definition}[Static Averaging Sequence (SAS)]\label{def: SAS}
Let $\hat c^{(k)} \in \hat L_2^\infty(I) \cap \hat L_2^\infty(\Gamma_y) \cap \mathbb R^{1\times \infty}$ for all $k \in \mathbb N$. The sequence $\left(\hat c^{(k)}:k \in \mathbb N\right)$ is called Static Averaging Sequence (SAS) if
\begin{align*}
    \lim_{k\to\infty} \hat c^{(k)}\left(\hat c^{(k)}\right)' = \lim_{k\to \infty} \norm{\hat c^{(k)}}_{\hat L_2^\infty(I)} =  0 .
\end{align*} 
\end{definition}
We denote the set of all static averaging sequences corresponding to $(y_{it})$ as 
{
\footnotesize
\begin{align*}
     \mathcal S(\Gamma_y) &:= \left \{ \left(\hat c^{(k)}\right) : \hat c^{(k)} \in \hat L_2^\infty(I) \cap \hat L_2^\infty(\Gamma_y) \cap \mathbb R^{1 \times \infty} \mbox{  } \forall k \in \mathbb N \mbox{  and  } \lim_{k\to \infty} \norm{\hat c^{(k)}}_{\hat L_2^\infty(I)} = 0 \right \} . 
\end{align*}
}
\begin{definition}[Statically Idiosyncratic]
A stochastic double sequence $(z_{it})$ is called statically idiosyncratic, if $\lim_{k\to\infty} \E\left[\hat c^{\{k\}} z_t^k\right]^2 = 0$ for all $(\hat c^{(k)}) \in \mathcal S (\Gamma_z)$ for all $t \in \mathbb Z$.
\end{definition}
The following Theorem has been stated for the dynamic case in \cite{forni2001generalized}, Theorem 1:
\begin{theorem}[Characterisation of Statically Idiosyncratic]\label{thm: charact stat idiosyncratic}
The following statements are equivalent: 
\begin{itemize}
    \item[(i)] A stationary stochastic double sequence $(z_{it})$ is statically idiosyncratic; 
    \item[(ii)] the first eigenvalue of the variance matrix is bounded, i.e.,
    \begin{align*}        
     \sup_{n \in \mathbb N} \mu_1(\Gamma_z^n) < \infty.
    \end{align*}
\end{itemize}
\end{theorem}
The proof is parallel to the dynamic case treated in \cite{forni2001generalized}, for details see \cite{gersing2023reconciling}.

The set of all random variables that can be written as the mean square limit of a static average defines a closed subspace of $\cspargel(y_{it}: \in \mathbb N)$ \citep[the proof is analogous to][Lemma 6]{forni2001generalized}. Denote by ``$\mslim$'' the mean square limit. 
\begin{definition}[Static Aggregation Space]
The space $\mathbb S_t(y) := \left\{ z_t : z_t = \mslim_{k\to\infty} \widehat c^{(k)} y_t, \ \mbox{where} \ \left(\widehat c^{(k)}\right) \in \mathcal S(\Gamma_y) \right\} \subset \cspargel (y_t)$ is called Static Aggregation Space at time $t$. 
\end{definition}
Note that the static aggregation space changes with $t \in \mathbb Z$ as it emerges from aggregations over the cross-section of $y_{it}$ - holding $t$ fixed.

We may suppose that $(y_{it})$ has a static factor structure:
\begin{assumption}[r-Static Factor Structure]\label{A: r-SFS struct}
The process $(y_{it})$ can be represented as 
\begin{align}
    y_{it} = \Lambda_i F_t + e_{it} = C_{it} + e_{it}, 
\end{align}
where $F_t$ is an $r\times 1$ dimensional process with $r < \infty$, $\E [F_t]=0$, $\E [F_t F_t'] = I_r$ for all $t\in\mathbb Z$, $\E [F_t e_{it}] = 0$ for all $i \in \mathbb N$ and $t\in\mathbb Z$, and
\begin{itemize}
    \item[(i)]  $\sup_{n\in\mathbb N} \mu_r(\Gamma^n_C) = \infty$;
    \item[(ii)]  $\sup_{n\in\mathbb N} \mu_1(\Gamma^n_e) < \infty$. 
\end{itemize}
\end{assumption}
We can compute \textit{static low rank approximations} (SLRA) of $y_t^n$ of rank $r$ via ``static'' principal components. For this consider the eigen-decomposition of the variance matrix: 
\begin{align}
    \Gamma_y^n = P_{(n)}' M_{(n)} P_{(n)} , \label{eq: eigendecomp Gamma}
\end{align}
where $P_{(n)} = P_{(n)}(\Gamma_y^n)$ is an orthogonal matrix of row eigenvectors and $M_{(n)} = M_{(n)}(\Gamma_y^n)$ is a diagonal matrix of the $r$-largest eigenvalues of $\Gamma_y^n$ sorted from largest to smallest. Denote by $p_{nj}$ the $j$-th row of $P_{(n)}$ and by $P_n := P_{nr}$ the sub-orthogonal matrix consisting of the first $r$ rows of $P_{(n)}$. Analogously we write $M_n$ to denote the $r\times r$ diagonal matrix of the largest $r$ eigenvalues of $\Gamma_y^n$. Recall that we associate $r$ with the number of divergent eigenvalues of $\Gamma_y^n$ (see A\ref{A: r-SFS struct}). Set
\begin{align}
     \mathcal K_{ni} &:= \mathcal K_{ni}(\Gamma_y^n) := p_{ni}' P_n \quad \mbox{the} \ i\mbox{-th row of} \ P_n' P_n \label{eq: SAS of SLRA} \\
    C_t^{[n]} &:= P_{nr}' P_{nr} y_t^n = P_n' P_n y_t^n \label{eq: static rank r approx} \\
    C_{it, n} &:= \mathcal K_{ni} y_t^n \quad \mbox{the} \ i\mbox{-th row of} \ C_t^{[n]} \label{eq: static rank r approx i-th} . 
\end{align}
Recall that $C_t^{[n]}$ is the best (with respect to mean squared error) possible approximation of $y_t^n$ by an $r$ dimensional vector of linear combinations of $y_{1t}, ..., y_{nt}$. The $r\times 1$ vector $P_n y_t^n$ are the first $r$ principal components of $y_t^n$ and provide such a vector of linear combinations, though not uniquely. We call $C_t^{[n]}$ the \textit{static rank $r$ approximation} of $y_t^n$ which is unique. 

\begin{assumption_output}\label{A: r-SFS struct output}
There exists a natural number $r < \infty$, such that
    \begin{itemize}
        \item[(i)] $\sup_{n \in \mathbb N} \mu_r(\Gamma_y^n) = \infty$;
        \item[(ii)] $\sup_{n \in \mathbb N} \mu_{r+1}(\Gamma_y^n) < \infty$.
    \end{itemize}
\end{assumption_output}

\begin{theorem}[\citealp{chamberlain1983arbitrage}]\label{thm: charact r-SFS}
Consider a stochastic double sequence $(y_{it})$ in $L_2(\mathcal P, \mathbb C)$, then:\\
1. AY\ref{A: r-SFS struct output} holds if and only if A\ref{A: r-SFS struct} holds; in this case \\
2. $C_{it} = \mslim_{n\to\infty} C_{it, n}$;\\
3. $r, C_{it}, e_{it}$ are uniquely identified from the output sequence $(y_{it})$;\\
4. $C_{it} = \proj\left(y_{it} \mid \mathbb S_t(y)\right)$
\end{theorem}
\noindent The proof is analogous to the proof of the dynamic case provided in \cite{forni2001generalized}. For a detailed exposition see \cite{gersing2023reconciling}. By Theorem \ref{thm: charact r-SFS}.1 the eigenvalue structure of $\Gamma_y^n$ in A\ref{A: r-SFS struct} is equivalent to the representation as a factor model. This justifies the name ``static factor sequence'':
\begin{definition}[$r$-Static Factor Sequence ($r$-SFS)]\label{def: r-SFS}
A stochastic double sequence $(y_{it})$ in $L_2(\mathcal P, \mathbb C)$ that satisfies A\ref{A: r-SFS struct} is called $r$-Static Factor Sequence, $r$-SFS.
\end{definition}
\end{document}